\documentclass[submission]{eptcs}
\providecommand{\event}{AFL 2014}
\usepackage{colortbl}
\usepackage{breakurl}

\usepackage{amsmath,amssymb}
\usepackage{graphicx} 

\usepackage{ifpdf}
\ifpdf
\fi

\newcommand{\resp}{respectively}
\newcommand{\dfa}{\textrm{DFA}}
\newcommand{\dbia}{\textrm{DBiA}}
\newcommand{\fwd}{{\textnormal{fwd}}}
\newcommand{\bwd}{{\textnormal{bwd}}}
\newcommand{\abs}[1]{\left\lvert{#1}\right\rvert}
\newcommand{\symd}{\mathop{\bigtriangleup}}

\newcommand*{\qed}{\raisebox{0.5ex}[0ex][0ex]{\framebox[1ex][l]{}}}
\newtheorem{theorem}{Theorem}
\newtheorem{lemma}[theorem]{Lemma}
\newtheorem{corollary}[theorem]{Corollary}
\newtheorem{example}[theorem]{Example}
\newenvironment{proof}{%
  \par\noindent
  {\rmfamily\itshape\mdseries Proof\/}:\hspace{\labelsep}\ignorespaces}%
  {\mbox{}\nolinebreak\hfill~%
  {\qed}
  \medbreak
}

\title{More Structural Characterizations of Some Subregular Language
  Families by Biautomata}

\author{Markus Holzer and Sebastian Jakobi
\institute{Institut f\"ur Informatik, Universit\"at Giessen,\\
  Arndtstr. 2, 35392 Giessen, Germany}
\email{\{holzer,sebastian.jakobi\}@informatik.uni-giessen.de}
}

\def\titlerunning{More Structural Characterizations of Some Subregular Language
  Families by Biautomata}
\def\authorrunning{Markus Holzer and Sebastian Jakobi}

\begin{document}

\maketitle

\begin{abstract}
  We study structural restrictions on biautomata such as, e.g.,
  acyclicity, permutation-freeness, strongly permutation-freeness, and
  orderability, to mention a few. We compare the obtained language
  families with those induced by deterministic finite automata with
  the same property. In some cases, it is shown that there is no
  difference in characterization between deterministic finite automata
  and biautomata as for the permutation-freeness, but there are also
  other cases, where it makes a big difference whether one considers
  deterministic finite automata or biautomata. This is, for instance,
  the case when comparing strongly permutation-freeness, which results
  in the family of definite language for deterministic finite
  automata, while biautomata induce the family of finite and co-finite
  languages. The obtained results nicely fall into the known landscape
  on classical language families.
\end{abstract}

\section{Introduction}
\label{sec:introduction}

The finite automaton is one of the first and most intensely
investigated computational model in theoretical computer science, see,
e.g.,~\cite{Mi67}. Its systematic study led to a rich and unified
theory of regular subfamilies such as, for example, finite languages
(are accepted by acyclic finite automata---here, except for
non-accepting sink states, self-loops on states count as cycles),
ordered languages (where the transitions of the accepting automata
preserve an order on the state set), and star-free languages or
non-counting languages (which can be described by regular like
expressions using only union, concatenation, and complement or
equivalently by permutation-free finite automata), to mention a
few. Relations between several subregular language families, such as
those mentioned above, are summarized in~\cite{Ha69a}. In particular,
an extensive study of star-free regular languages can be found
in~\cite{McNaPa71}. Even nowadays the study of subregular language
families from different perspectives such as, for instance, algebra,
logic, descriptional, or computational complexity, is a vivid area of
research.

Recently, an alternative automaton model to the deterministic finite
automaton (\dfa), the so called \emph{biautomaton}
(\dbia)~\cite{KlPo12} was introduced. Roughly speaking, a biautomaton
consists of a \emph{deterministic} finite control, a read-only input
tape, and two reading heads, one reading the input from left to right
(forward transitions), and the other head reading the input from right
to left (backward transitions). Similar two-head finite automata
models were introduced, e.g., in~\cite{CDJM13a,Lo07,Ro67}. An input
word is accepted by a biautomaton, if there is an accepting
computation starting the heads on the two ends of the word meeting
somewhere in an accepting state.  Although the choice of reading a
symbol by either head is nondeterministic, a deterministic outcome of
the computation of the biautomaton is enforced by two properties: (i)
The heads read input symbols independently, i.e., if one head reads a
symbol and the other reads another, the resulting state does not
depend on the order in which the heads read these single letters. (ii)
If in a state of the finite control one head accepts a symbol, then
this letter is accepted in this state by the other head as well.
Later we call the former property the $\diamond$-property and the
latter one the $F$-property.  In~\cite{KlPo12} and a series of
forthcoming papers~\cite{HoJa13c,HoJa13a,JiKl12,KlPo12a} it was shown
that biautomata share a lot of properties with ordinary finite
automata. For instance, as minimal \dfa s, also minimal \dbia s are
unique up to isomorphism~\cite{Ar69,KlPo12}.

Now the question arises, which structural characterizations of
subregular language families of \dfa s carry over to biautomata. Let
us give an example which involves partially ordered automata. A \dfa\
with state set~$Q$ and input alphabet~$\Sigma$ is \emph{partially
  ordered}, if there is a (partial) order~$\leq$ on~$Q$ such that
$q\leq \delta(q,a)$, for every $q\in Q$
and~$a\in\Sigma$. In~\cite{BrFi80} it was shown that partially ordered
\dfa s characterize the family of $\mathcal{R}$-trivial regular
languages, that is, a regular language~$L$ is $\mathcal{R}$-trivial if
for its syntactic monoid~$M_L$, the assumption $sM_L=tM_L$ implies
$s=t$, for all $s,t\in M_L$. For the definition of the syntactic
monoid of a regular language we refer to~\cite{Ar69}. Adapting the
definition of being partially ordered literally to \dbia s results in
a characterization of the family of $\mathcal{J}$-trivial regular
languages~\cite{KlPo12a,KlPo12}---originally the authors
of~\cite{KlPo12} speak of acyclic biautomata instead, since loops are
not considered as cycles there; we think that the term \emph{partially
  ordered} is more suitable in this context. Here a regular
language~$L$ is $\mathcal{J}$-trivial if for its syntactic
monoid~$M_L$, the assumption $M_LsM_L=M_LtM_L$ implies $s=t$, for all
$s,t\in M_L$. Note that a language is $\mathcal{J}$-trivial regular if
and only if it is piecewise testable~\cite{Si75}. A
language~$L\subseteq\Sigma^*$ is \emph{piecewise testable} if it is a
finite Boolean combination of languages of the form
$\Sigma^*a_1\Sigma^*a_2\Sigma^*\ldots\Sigma^*a_n\Sigma^*$, where
$a_i\in\Sigma$ for $1\leq i\leq n$. We can also ask whether a transfer
of conditions can be done the other way around from \dbia s to \dfa
s. This is not that obvious, since structural properties on \dbia s
may involve conditions on the forward and backward transitions. For
instance, in~\cite{HoJa13c} it was shown that biautomata, where for
every state and every input letter the forward and the backward
transition go to the same state, characterize the family of
commutative regular languages. A regular language $L\subseteq\Sigma^*$
is \emph{commutative} if for all words $u,v\in\Sigma^*$ and letters
$a,b\in\Sigma$ we have $uabv\in L$ if and only if $ubav\in
L$. Obviously, this condition can be used also to give a structural
characterization of commutative regular languages on \dfa s, namely
that for every state~$q$ and letters $a,b\in\Sigma$ the finite state
device satisfies $\delta(\delta(q,a),b)=\delta(\delta(q,b),a)$. This
is the starting point of our investigations.

We study structural properties of \dfa s appropriately adapted to
\dbia s, since up to our knowledge most classical properties from the
literature on finite automata were not studied for \dbia s yet.  Our
investigation is started in Section~\ref{sec:permutation-automata}
with automata which transition functions induce permutations on the
state set. Originally permutation \dfa s were introduced
in~\cite{Th68a}. We show that both types of finite state machines,
permutation \dfa s and permutation \dbia s are equally powerful. Thus,
an alternative characterization of the family of $p$-regular languages
in terms of \dbia s is obtained. Next we take a closer look on quite
the opposite of permutation automata, namely on permutation-free
devices---see Section~\ref{sec:permutation-free}.  A special case of a
permutation-free automaton is an acyclic (expect for sink states)
one. It is easy to see that acyclic \dfa s as well as \dbia s
characterize the family of finite languages.
An important subregular language family, which can be obtained from
finite languages by finitely many applications of concatenation,
union, and complementation with respect to the underlying alphabet, is
the class of star-free languages.  It obeys a variety of different
characterizations~\cite{McNaPa71}, one of them are permutation-free
\dfa s. We show that permutation-free \dbia s characterize the
star-free languages, too. For strongly permutation-free automata,
which are automata that are permutation-free and where also the
identity permutation is forbidden, we find the first significant
difference of \dfa s and \dbia s. While for \dfa s this property
characterizes the family of definite languages, \dbia s describe only
finite or co-finite languages. A language~$L\subseteq\Sigma^*$ is
\emph{definite}~\cite{PRS63} if and only if $L=L_1\cup\Sigma^*L_2$,
for some finite languages~$L_1$ and~$L_2$.  Moreover, we find a
relation between strongly permutation-free automata, and automata
where all states are almost-equivalent---the notion of
almost-equivalence was introduced in~\cite{BGS09}.  Then in
Section~\ref{sec:ordered-automata} we continue our investigation with
another important subfamily of star-free languages, namely ordered
languages~\cite{ShTh74}. A \dfa\ with state set~$Q$ and input
alphabet~$\Sigma$ is \emph{ordered} if there is a total order~$\leq$
on the state set~$Q$ such that $p\leq q$ implies $\delta(p,a)\leq
\delta(q,a)$, for every $p,q\in Q$ and $a\in\Sigma$. The family of
ordered languages lies strictly in-between the family of finite and
the family of star-free languages. Appropriately adapting this
definition to biautomata results in a language class, which we call
the family of bi-ordered languages, that is a proper superset of the
family of finite and co-finite languages and a strict subset of the
family of ordered languages. Moreover, it is shown that there is a
subtle difference whether the order condition is applied to automata
in general or to minimal devices only. In the next to last section we
take a closer look on non-exiting and non-returning machines. It is
well known that non-exiting \dfa s characterize the family of
prefix-free languages, while non-returning automata are related to
suffix-free languages. We show that every biautomaton which is
non-exiting must also be non-returning (unless it accepts the empty
language), and that non-exiting minimal \dbia s characterize the
family of circumfix-free languages. For non-returning minimal \dbia s
we prove that the induced language family is a strict subset of the
family of non-returning minimal \dfa s languages. The obtained results
are summarized in Table~\ref{tab:results}.
\def\shade{\cellcolor[gray]{0.75}}%
\begin{table}
  \centering
  \begin{tabular}{|l||c|c|}\hline\hline
 & \multicolumn{2}{c|}{Automata type}\\\cline{2-3}
Property & \dfa  s & \dbia s \\\hline\hline
permutation & $p$-regular & \shade $p$-regular\\\hline
permutation-free & star-free & \shade star-free\\\hline 
ordered & ordered & \shade $\mbox{FIN}\cup\mbox{co-FIN}\subset\cdot\subset\mbox{ORD}$ \\\hline
partially ordered & $\mathcal{R}$-trivial & $\mathcal{J}$-trivial\\\hline 
strongly permutation-free & \shade definite & \shade finite and co-finite\\\hline
acyclic; self-loops are cycles & finite & \shade finite\\\hline
non-exiting & prefix-free & \shade circumfix-free\\\hline
non-returning & strict superset of suffix-free & \shade strict subset of \mbox{non-returning \dfa s} \\\hline\hline
  \end{tabular}
  \caption{Comparison of the results on structural properties on \dfa s and \dbia s and their induced language families (shading represents results obtained in this paper); here $\mbox{FIN}$ refers to the family of finite languages, $\mbox{co-FIN}$ to the family of co-finite languages, and $\mbox{ORD}$ to the family of ordered languages.}
  \label{tab:results}
\end{table}
In the last section we briefly discuss our findings and give some
hints on future research directions on the subject under
consideration.

\section{Preliminaries}
\label{sec:preliminaries}

A \emph{deterministic finite automaton} (\dfa) is a quintuple
$A=(Q,\Sigma,\delta,q_0,F)$, where~$Q$ is the finite set of
\emph{states}, $\Sigma$ is the finite set of \emph{input symbols},
$q_0\in Q$ is the \emph{initial state}, $F\subseteq Q$ is the set of
\emph{accepting states}, and $\delta\colon Q\times\Sigma\to Q$ is the
\emph{transition function}.  As usual, the transition
function~$\delta$ can be recursively extended to $\delta\colon
Q\times\Sigma^*\to Q$.  The \emph{language accepted} by~$A$ is defined
as $L(A) =\{\,w\in\Sigma^*\mid \delta(q_0,w)\in F\,\}$.

A \emph{deterministic biautomaton} (\dbia) is a sixtuple
$A=(Q,\Sigma,\cdot,\circ,q_0,F)$, where~$Q$, $\Sigma$, $q_0$, and~$F$
are defined as for \dfa s, and where~$\cdot$ and~$\circ$ are mappings
from~$Q\times\Sigma$ to~$Q$, called the \emph{forward} and
\emph{backward transition function}, respectively.  It is common in
the literature on biautomata to use an infix notation for these
functions, i.e., writing~$q\cdot a$ and~$q\circ a$ instead
of~$\cdot(q,a)$ and $\circ(q,a)$.  Similar as for the transition
function of a \dfa, the forward transition function~$\cdot$ can be
extended to $\cdot\colon Q\times\Sigma^* \to Q$ by 
$q\cdot \lambda = q$ and $q\cdot av = (q\cdot a)\cdot v$,
for all states~$q\in Q$, symbols~$a\in\Sigma$, and
words~$v\in\Sigma^*$. Here~$\lambda$ refers to the \emph{empty
  word}. The extension of the backward transition function~$\circ$ to
$\circ\colon Q\times\Sigma^* \to Q$ is defined as follows:
$q\circ \lambda = q$ and $q\circ va = (q\circ a)\circ v$, 
for all states~$q\in Q$, symbols~$a\in\Sigma$, and
words~$v\in\Sigma^*$.  Notice that~$\circ$ consumes the input from
right to left, hence the name backward transition function.

The \dbia~$A$ \emph{accepts} a word~$w\in\Sigma^*$ if there are
words $u_i,v_i\in\Sigma^*$, for~$1\leq i\leq k$, such that~$w$ can be
written as~$w=u_1u_2\dots u_k v_k \dots v_2 v_1$, and
\[((\dots((((q_0\cdot u_1)\circ v_1)\cdot u_2)\circ v_2)\dots)\cdot
u_k)\circ v_k \in F. \]
The language accepted by~$A$ is $L(A) = \{\, w\in\Sigma^* \mid
\text{$A$ accepts $w$}\,\}$.

The \dbia~$A$ has the \emph{$\diamond$-property}, if $(q\cdot a)\circ
b = (q\circ b)\cdot a$, for all~$a,b\in\Sigma$, and~$q\in Q$, and it
has the \emph{$F$-property}, if for all~$q\in Q$ and~$a\in\Sigma$ it
is~$q\cdot a \in F$ if and only if~$q\circ a \in F$.
The biautomata as introduced in~\cite{KlPo12} always had to satisfy
both these properties, while in~\cite{HoJa13c,HoJa13a} also biautomata
that lack one or both of these properties, as well as nondeterministic
biautomata were studied.
Throughout the current paper, when writing of biautomata, or \dbia s,
we always mean deterministic biautomata that satisfy both the
$\diamond$-property, and the $F$-property, i.e., the model as
introduced in~\cite{KlPo12}.
For such biautomata the following it is known from the
literature~\cite{HoJa13c,KlPo12}:
\begin{itemize}
\item $(q\cdot u)\circ v = (q\circ v)\cdot u$, for all states~$q\in Q$
  and words~$u,v\in\Sigma^*$,
\item $(q\cdot u)\circ vw \in F$ if and only if $(q\cdot uv)\circ w\in
  F$, for all states~$q\in Q$ and words~$u,v,w\in\Sigma^*$.
\end{itemize}
From this one can conclude that for all words $u_i,v_i\in\Sigma^*$,
with~$1\leq i\leq k$, it is 
\[((\dots((((q_0\cdot u_1)\circ v_1)\cdot u_2)\circ v_2)\dots)\cdot
u_k)\circ v_k \in F \]
if and only if
\[ q_0 \cdot u_1u_2\dots u_kv_k \dots v_2 v_1 \in F. \]
Therefore, the language accepted by a \dbia~$A$ can as well be defined
as $L(A) = \{\, w\in\Sigma^* \mid q_0\cdot w \in F\,\}$.

Let~$A$ be \dfa\ or a \dbia\ with state set~$Q$. We say that a state
$q\in Q$ is a \emph{sink state} if and only if all outgoing transition
(regardless whether they are forward or backward transitions) are
self-loops only. Note, that in particular, one can distinguish between
accepting and non-accepting sink states.

In the following we define the two \dfa s contained in a \dbia, which
accept the language, and the reversal of the language accepted by the
biautomaton.
Let $A=(Q,\Sigma,\cdot,\circ,q_0,F)$ be a \dbia.  We denote
by~$Q_\fwd$ the set of all states reachable from~$q_0$ by only using
forward transitions, and denote the set of states reachable by only
using backward transitions by~$Q_\bwd$, i.e.,
\begin{align*}
  Q_\fwd = \{\, q\in Q \mid \exists u\in\Sigma^*: q_0\cdot u = q\,\}\quad\mbox{and}\quad
  Q_\bwd = \{\, q\in Q \mid \exists v\in\Sigma^*: q_0\circ v = q\,\}.
\end{align*}
Now we define the \dfa\ $A_\fwd =
(Q_\fwd,\Sigma,\delta_\fwd,q_0,F_\fwd)$, with $F_\fwd = Q_\fwd \cap
F$, and $\delta_\fwd(q,a) = q\cdot a$, for all states~$q\in
Q_\fwd$ and symbols~$a\in\Sigma$.  Similarly, we define the \dfa\
$A_\bwd = (Q_\bwd,\Sigma,\delta_\bwd,q_0,F_\bwd)$, with $F_\bwd =
Q_\bwd \cap F$, and $\delta_\bwd(q,a)=q\circ a$, for all~$q\in Q$
and~$a\in\Sigma$.
One readily sees that~$L(A_\fwd)=L(A)$.  Moreover, since~$q\circ uv=
(q\circ v)\circ u$, one can also see~$L(A_\bwd)=L(A)^R$.
It is shown in~\cite{HoJa14} that if~$A$ is a minimal biautomaton,
then the two \dfa s~$A_\fwd$ and~$A_\bwd$ are minimal, too.

\section{Permutation Automata}
\label{sec:permutation-automata}

First we study automata where every input induces a permutation on the
state set.  Such finite automata were defined in~\cite{Th68a}.  A
\dfa\ $A=(Q,\Sigma,\delta,q_0,F)$ is a \emph{permutation \dfa} if
$\delta(p,a) = \delta(q,a)$ implies~$p=q$, for all~$p,q\in Q$
and~$a\in\Sigma$.  A regular language is \emph{$p$-regular} if it is
accepted by a permutation \dfa.
We give a similar definition for biautomata: a biautomaton
$A=(Q,\Sigma,\cdot,\circ,q_0,F)$ is a \emph{permutation biautomaton}
if for all~$p,q\in Q$ and~$a\in \Sigma$ we have that $p\cdot a =
q\cdot a$ implies~$p=q$, and also $p\circ a = q\circ a$ implies~$p=q$.

We will see that a language is $p$-regular if and only if it is
accepted by a permutation biautomaton.  Before we can show this, we
describe a useful technique to construct a biautomaton from finite
automata.
In~\cite{KlPo12} a construction of a biautomaton from a given \dfa~$A$
is described, that uses a cross-product construction of~$A$ with the
power-set automaton of the reversal of~$A$.
In the following we describe how a biautomaton can be constructed from
two arbitrary \dfa s accepting a regular language and its reversal.

Let~$L\subseteq\Sigma^*$ be a regular language, and for~$i\in\{1,2\}$
let~$A_i=(Q_i,\Sigma,\delta_i,q_0^{(i)},F_i)$ be \dfa s
with~$L(A_1)=L$, and~$L(A_2)=L^R$.  Further, for all states~$p\in Q_1$
let~$u_p$ be some word with~$\delta_1(q_0^{(1)},u_p)= p$, and
similarly for~$q\in Q_2$ let~$v_q$ be a word
with~$\delta_2(q_0^{(2)},v_q)= q$.
Then define the automaton $B_{A_1\times A_2} =
(Q,\Sigma,\cdot,\circ,q_0,F)$ with state set~$Q=Q_1\times Q_2$,
initial state~$q_0=(q_0^{(1)},q_0^{(2)})$, accepting states
$F=\{\,(p,q)\in Q \mid u_pv_q^R \in L\,\}$, and where for
all~$(p,q)\in Q$ and~$a\in \Sigma$ we have $(p,q)\cdot a =
(\,\delta_1(p,a),\,q\,)$, and $(p,q)\circ a =
(\,p,\,\delta_2(q,a)\,)$.  The following lemma proves the correctness
of this construction.

\begin{lemma}\label{lem:cross-product-dbia}
  For~$i\in\{1,2\}$ let~$A_i=(Q_i,\Sigma,\delta_i,q_0^{(i)},F_i)$ be
  \dfa s with~$L(A_1)=L$, and~$L(A_2)=L^R$.  Then~$B_{A_1\times A_2}$
  is a deterministic biautomaton, such that~$L(B_{A_1\times A_2}) =
  L$.
\end{lemma}

Besides its usefulness for our result on permutation biautomata, this
construction is also of relevance from a descriptional complexity
point of view.
Using the construction from~\cite{KlPo12} on an $n$-state \dfa\ yields
a biautomaton with~$n\cdot 2^n$ states.  In fact, a precise analysis
in~\cite{JiKl12} that uses a similar construction as in~\cite{KlPo12}
proves a tight bound of $n\cdot 2^n - 2(n-1)$ states for converting an
$n$-state \dfa\ into an equivalent biautomaton.  However, this bound
only takes into account the state complexity of the original
language~$L$, but not the state complexity of~$L^R$.  If the state
complexity of~$L^R$ much smaller than~$2^n$ then the bound
from~\cite{JiKl12} is far off the number of states of the minimal
biautomaton for~$L$.  Using our construction, we can deduce an upper
bound of~$n\cdot m$ for the number of states of a biautomaton for the
language~$L$, if~$n$ is the state complexity of~$L$, and~$m$ is the
state complexity of~$L^R$.

Now we show our result on permutation automata.

\begin{theorem}
  A language is $p$-regular if and only if it is accepted by some
  permutation biautomaton.
\end{theorem}

\begin{proof}
  If~$A$ is a permutation biautomaton, then~$A_\fwd$ is a permutation
  \dfa, hence~$L(A)$ is $p$-regular.  For the reverse implication
  let~$L$ be some $p$-regular language over the alphabet~$\Sigma$.
  Then~$L^R$ is $p$-regular, too~\cite{Th68a}, so there are
  permutation \dfa s~$A_i=(Q_i,\Sigma,\delta_i,q_0^{(i)},F_i)$,
  for~$i=1,2$, that~$L=L(A_1)$, and~$L^R=L(A_2)$.  Using the
  cross-product construction from Lemma~\ref{lem:cross-product-dbia},
  we obtain the biautomaton $B=B_{A_1\times A_2}$.  Recall that the
  states of~$B$ are of the form~$(p,q)$, with~$p\in Q_1$ and~$q\in
  Q_2$, and the transitions are defined such that $(p,q)\cdot a =
  (\delta_1(p,a),q)$, and $(p,q)\circ a = (p,\delta_2(q,a))$, for all
  states~$(p,q)$ and symbols~$a\in\Sigma$.  We will show in the
  following that~$B$ is a permutation automaton.  Therefore
  let~$(p,q)$ and~$(p',q')$ be two states of~$B$, and~$a\in\Sigma$.
  If $(p,q)\cdot a = (p',q')\cdot a$ then $(\delta(p,a),q) =
  (\delta(p',a),q')$, which implies~$\delta(p,a)=\delta(p',a)$
  and~$q=q'$.  Since~$A_1$ is a permutation \dfa, we also
  obtain~$p=p'$, hence~$(p,q)=(p',q')$.  With a similar reasoning,
  using the permutation property of~$A_2$, we see that also
  $(p,q)\circ a = (p',q')\circ a$ implies~$(p,q)=(p',q')$,
  therefore~$B$ is a permutation biautomaton.
\end{proof}

\section{Permutation-Free Automata}
\label{sec:permutation-free}

An important subregular language family is the class of star-free
languages. A language is star-free if it can be obtained from finite
languages by finitely many applications of concatenation, union, and
complementation with respect to the underlying alphabet.  For the
class of finite languages we have the following obvious theorem, which
we state without proof.

\begin{theorem}
  A language is finite (co-finite, \resp) if and only if its minimal
  biautomaton is acyclic except for non-accepting (accepting, \resp)
  sink states; self-loops count as cycles.\hfill\qed
\end{theorem}

The class of star-free languages obeys a variety of different
characterizations~\cite{McNaPa71}, one of them being the following: a
regular language is star-free if its minimal \dfa\ is
permutation-free.  A \dfa\ $A=(Q,\Sigma,\delta,q_0,F)$ is
\emph{permutation-free} if there is no word~$w\in\Sigma^*$ such that
the mapping $q\mapsto \delta(q,w)$, for all~$q\in Q$, induces a
\emph{non-trivial} permutation, i.e., a permutation different from the
identity permutation, on some set~$P\subseteq Q$.
Now the question arises whether a similar condition for biautomata
also yields a characterization of the star-free languages.  We feel
that the following definition of permutation-freeness is a natural
extension from the corresponding definition for \dfa s.  We say that a
biautomaton $A=(Q,\Sigma,\cdot,\circ,q_0,F)$ is
\emph{permutation-free} if there are no words~$u,v\in\Sigma^*$, such
that the mapping $q\mapsto (q\cdot u)\circ v$ induces a non-trivial
permutation on some set of states~$P\subseteq Q$.  It turns out that
with this definition, permutation-free biautomata indeed characterize
the star-free languages.  We will later discuss some other possible
definitions.  Before we show our result on permutation-free
biautomata, we prove the following lemma which helps us to relate
permutations in biautomata to permutations in \dfa s.

\begin{lemma}\label{lem:permutation-lemma}
  Let~$Q$ be a finite set and~$\pi_1,\pi_2\colon Q\to Q$ be two
  mappings satisfying $\pi_1(\pi_2(q)) = \pi_2(\pi_1(q))$ for all
  $q\in Q$.  If there exists a subset~$P\subseteq Q$ such that the
  mapping $\pi\colon P\to P$ defined by $\pi(p) = \pi_2(\pi_1(p))$ is
  a non-trivial permutation on~$P$, then there exists a
  subset~$P'\subseteq Q$ and an integer~$d\geq 1$ such that~$\pi_1^d$
  or~$\pi_2^d$ is a non-trivial permutation on~$P'$.
\end{lemma}

\begin{proof}
  Consider the sequence of sets $\pi_1^0(P), \pi_1^1(P), \pi_1^2(P),
  \ldots \subseteq Q$.  Since~$Q$ is a finite set, the number of
  different sets~$\pi_1^j(P)$, for~$j\geq 0$, is finite.  Thus, there
  must be integers~$m,d\geq 1$ such that the sets $\pi_1^0(P),
  \pi_1^1(P), \dots \pi_1^{m+d-1}(P)$ are pairwise distinct,
  and~$\pi_1^{m+d}(P) = \pi_1^m(P)$.  Since~$\pi =
  \pi_1\pi_2=\pi_2\pi_1$ is a permutation on~$P$, we obtain
  \[ P=\pi^{m+d}(P) = \pi_2^{m+d}(\pi_1^{m+d}(P)) = \pi_2^d ( \pi_2^m
  ( \pi_1^m (P))) = \pi_2^d (P), \]
  which shows that~$\pi_2^d$ is a permutation on~$P$.  It follows that
  also~$\pi_1^d$ must be a permutation on~$P$.  If one of these is a
  non-trivial permutation we are done.  
  Therefore assume that both~$\pi_1^d$ and~$\pi_2^d$ are the identity
  permutation on~$P$.  In this case it must be~$d\geq 2$ because
  otherwise the permutation~$\pi=\pi_1\pi_2$ would be trivial.  Then
  the two sets~$P$ and~$\pi_1(P)$ are different, so there is an
  element~$q\in P$ with~$\pi_1(q)\neq q$.  On the other hand~$q$ must
  satisfy~$\pi_1^d(q) = q$.  Therefore the mapping~$\pi_1$ is a
  permutation on the set
  $P'=\{\pi_1^0(p),\pi_1^1(p),\dots,\pi_1^{d-1}(p)\}$, and it is
  non-trivial because~$\pi_1^0(p) = p \neq \pi_1^1(p)$.
\end{proof}

Now we can show the following characterization of star-free languages
in terms of permutation-free biautomata.

\begin{theorem}\label{thm:star-free-permutation-free}
  A language is star-free if and only if its \emph{minimal}
  biautomaton is permutation-free.
\end{theorem}

\begin{proof}
  Clearly, if~$A$ is a minimal biautomaton that is permutation-free,
  then also the contained minimal \dfa~$A_\fwd$ is permutation-free,
  too.  Therefore the language~$L(A)$ is star-free.  

  For proving the reverse implication let
  $A=(Q,\Sigma,\cdot,\circ,q_0,F)$ be a minimal biautomaton that is
  \emph{not} permutation-free.  Then there are words~$u,v\in\Sigma^*$,
  and a set of states~$P\subseteq Q$, with~$\abs P \geq 2$, such that
  the mapping~$\pi\colon Q\to Q$ defined as $\pi(q) = (q\cdot u)\circ
  v$ induces a non-trivial permutation on~$P$.  Notice that the
  $\diamond$-property of the biautomaton~$A$ implies
  $\pi(q)=\pi_1(\pi_2(q))=\pi_2(\pi_1(q))$, for $\pi_1(q) = q\cdot u$,
  and $\pi_2(q) = q\circ v$.  Therefore we can use
  Lemma~\ref{lem:permutation-lemma}, and obtain an integer~$d\geq 1$
  such that~$\pi_1^d$ or~$\pi_2^d$ induces a non-trivial permutation
  on some subset~$P'\subseteq Q$.  If~$\pi_1^d$ is non-trivial, then
  the word~$u^d$ induces a non-trivial permutation in the minimal
  \dfa~$A_\fwd$, which in turn means that the language~$L(A)$ is
  \emph{not} star-free.  Otherwise, the mapping~$\pi_2^d$ is
  non-trivial, and the word~$v^d$ induces a non-trivial permutation on
  the minimal \dfa~$A_\bwd$, which means that the language~$L(A)^R$ is
  not star-free.  Since the class of star-free languages is closed
  under reversal, we again conclude that the language~$L(A)$ cannot be
  star-free in this case.
\end{proof}

In the above definition of permutation-free biautomata, from all
states in the permutation induced by the word~$uv$, the prefix~$u$
must be read with forward transitions and the suffix~$v$ with backward
transitions.  
One could also think of other kinds of permutations in biautomata, and
we shortly discuss two different such notions in the following.  Since
a permutation is composed of cycles, we describe the types of cycles.
Let $P=\{p_0, p_1, \dots, p_{k-1}\}$ be some set of states of a
biautomaton~$A$ and~$w$ be some non-empty word over the input alphabet
of~$A$.
\begin{itemize}
\item We say that~$w$ induces a \emph{word-cycle} on~$P$ if for $0\leq
  i\leq k-1$ we have $p_{(i+1)\bmod k} = (p_i\cdot u_i)\circ v_i$, for
  some words~$u_i$ and~$v_i$, with~$u_iv_i=w$.
\item We say that~$w$ induces a \emph{graph-cycle} on~$P$ if
  $w=a_1a_2\dots a_n$ and we have
  \[ p_{(i+1)\bmod k} = (\dots((p_i \bullet_{i,1} a_1) \bullet_{i,2}
  a_2)\dots)\bullet_{i,n} a_n, \]
  for $0\leq i\leq k-1$, where~$\bullet_{i,j}\in \{\cdot,\circ\}$,
  for~$1\leq j\leq n$. The intuition behind the definition of a
  graph-cycle is that the word~$w$ specifies the sequence of
  transitions (regardless whether they are forward or backward
  transitions) that are taken during the course of the computation.
\end{itemize}

Notice that if a biautomaton has a permutation as defined before
Theorem~\ref{thm:star-free-permutation-free}, then it also has a
word-cycle, and also a graph-cycle.  Hence, if a language is accepted
by a biautomaton that has no word-cycle or by a biautomaton that has
no graph-cycle, then it is star-free.  However, the converse is not
true, as the following example shows.

\begin{example}
  Let $A=(Q,\Sigma,\cdot,\circ,q_0,F)$ be the \emph{minimal} \dbia\ for the
  language~$L=\{aab,bab\}^*$.  The biautomaton~$A$ is depicted in
  Figure~\ref{fig:permutation-free-with-word-graph-cycle}---solid
  arrows denote forward transitions by~$\cdot$, and dashed arrows
  denote backward transitions by~$\circ$.
  \begin{figure}
    \centering
    \includegraphics[scale=.9]{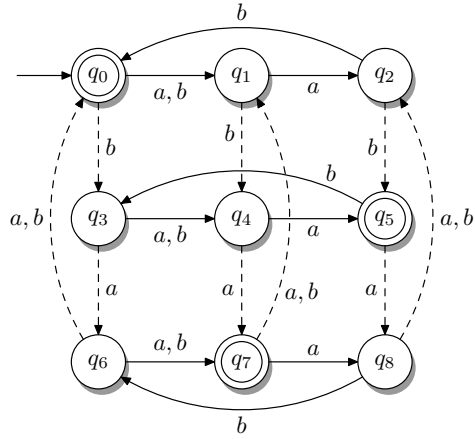}
    \caption{The permutation-free biautomaton~$A$ that has a
      word-cycle and a graph-cycle.}
    \label{fig:permutation-free-with-word-graph-cycle}
  \end{figure}
  By inspecting the \dfa~$A_\fwd$, consisting of the
  states~$q_0,q_1,q_2$, and the non-accepting sink state, which is not
  shown, one can see that~$A_\fwd$ is permutation-free.  Therefore the
  language~$L$ is star-free, and also the biautomaton~$A$ must be
  permutation-free.  However, the word~$ab$ induces a word-cycle on
  the states~$q_0$,~$q_1$, and~$q_6$ because $q_0\circ ab = q_6$,
  $(q_6\cdot a)\circ b = q_1$, and $q_1\cdot ab = q_0$.  Further, the
  word~$ab$ also induces a graph-cycle on the states~$q_0$,~$q_4$,
  and~$q_3$ because $(q_0\cdot a)\circ b = q_4$, $(q_4\cdot a)\cdot b
  = q_3$, and $(q_3\circ a)\circ b = q_0$.
\end{example}

\subsection{Strongly Permutation-Free Automata}
\label{sec:strongly-perm-free}

A permutation-free automaton does not contain any \emph{non-trivial}
permutation, but it may contain the identity permutation.  We may also
forbid the identity permutation, which leads to the following
definitions.
A \dfa~$A=(Q,\Sigma,\delta,q_0,F)$ is \emph{strongly permutation-free}
if there is no non-empty word~$w\in\Sigma^+$, such that the
mapping~$q\mapsto \delta(q,w)$ induces a permutation on some
set~$P\subseteq Q$, with~$\abs P\geq 2$.
Similarly, a biautomaton $A=(Q,\Sigma,\cdot,\circ,q_0,F)$ is
\emph{strongly permutation-free} if there are no
words~$u,v\in\Sigma^*$, with~$uv\in\Sigma^+$, such that the
mapping~$q\mapsto (q\cdot u)\circ v$ induces a permutation on some
set~$P\subseteq Q$, with~$\abs P\geq 2$.
In these definitions we require the words~$w$ and~$uv$ to be non-empty
because the empty-word always induces the identity permutation on all
sets of states.  Further we only consider subsets~$P\subseteq Q$
with~$\abs P\geq 2$ because every \dfa\ and every biautomaton must
contain a (maybe identity) permutation on a set~$P\subseteq Q$
with~$\abs P \geq 1$, since by the pigeon hole principle there is a
state which is repeatedly visited by only reading the letter~$a$ long
enough.

Before we study which languages are accepted by strongly
permutation-free automata, we give some further definitions which turn
out to be related to strongly permutation-freeness.
Let~$A=(Q,\Sigma,\delta,q_0,F)$ be a \dfa, and~$w\in \Sigma^+$ some
non-empty word.  A state~$q\in Q$ is called a \emph{$w$-attractor}
in~$A$, if for all states~$p\in Q$ there is an integer~$k>0$ such
that~$\delta(p,w^k)=q$.
For a biautomaton $A=(Q,\Sigma,\cdot,\circ,q_0,F)$ and
words~$u,v\in\Sigma^*$, with~$\abs{uv}\geq 1$, we denote
by~$\pi_{u,v}$ the mapping $p\mapsto (p\cdot u)\circ v$.  Now a
state~$q\in Q$ is a \emph{$(u,v)$-attractor} in~$A$ if for all~$p\in
Q$ there is an integer~$k>0$ such that $\pi_{u,v}^k(p) = q$.  Notice
that due to the $\diamond$-property of~$A$ the
condition~$\pi_{u,v}^k(p) = q$ can also be written as $(p\cdot
u^k)\circ v^k = q$.
Next we recall the definition of almost-equivalence.  Two
languages~$L_1$ and~$L_2$ are \emph{almost-equivalent} ($L_1\sim L_2$)
if their symmetric difference $L_1\symd L_2 = (L_1\setminus L_2) \cup
(L_2\setminus L_1)$ is finite.  This notion naturally transfers to
states as follows.  For a state~$q$ of some \dfa\ or biautomaton~$A$
we denote by~$L_A(q)$ the language accepted by the automaton~${}_qA$
which is obtained from~$A$ by making state~$q$ its initial state.  The
language~$L_A(q)$ is also called the \emph{right language} of~$q$. Now
two states~$p$ and~$q$ of a \dfa\ or biautomaton~$A$ are
\emph{almost-equivalent} ($p\sim q$) if their \emph{right
  languages}~$L_A(p)$ and~$L_A(q)$ are almost-equivalent.  We
write~$p\equiv q$, if~$p$ and~$q$ are \emph{equivalent}, i.e.,
if~$L_A(p)=L_A(q)$.
Our next theorem connects the notions of almost-equivalence,
$w$-attractors, and strongly permutation-freeness for \dfa s, and
shows that these conditions can be used to characterize the class of
definite languages---a language~$L$ over an alphabet~$\Sigma$ is
\emph{definite} if there are finite languages~$L_1$ and~$L_2$
over~$\Sigma$ such that $L = L_1\cup \Sigma^*L_2$.
Interestingly the relation between definite languages and
almost-equivalence was already studied in~\cite{PRS63}, long before
the notion of almost-equivalence became popular
in~\cite{BGS09}---in~\cite{PRS63} the used form of equivalence was
not called ``almost-equivalence,'' but simply ``equivalent.''
Moreover, the relation between definite languages and strongly
permutation-free automata was independently shown
in~\cite{BrLi12}---compare also with~\cite[Exercise~28 of Chapter~4
and Exercise~13 of Chapter~5]{McNaPa71}. 

\begin{theorem}\label{thm:dfa-characterization-almost-definite}
  Let~$A=(Q,\Sigma,\delta,q_0,F)$ be some \emph{minimal} \dfa, then the
  following statements are equivalent:
  \begin{enumerate}
  \item All states in~$Q$ are pairwise
    almost-equivalent.  \label{item:charact-almost}
  \item For all words~$w\in\Sigma^+$, there is a~$w$-attractor
    in~$A$.  \label{item:charact-attractor}
  \item $A$ is strongly
    permutation-free.  \label{item:charact-permutation}
  \item $L(A)$ is a definite language.\label{item:charact-definite}
  \end{enumerate}
\end{theorem}

Now we turn to biautomata, where we will see that the equivalences
between the first three conditions of
Theorem~\ref{thm:dfa-characterization-almost-definite} also hold in
the setting of biautomata.  However, we will see that the language
class related to these conditions is different.
Before we come to this result we recall the following lemma
from~\cite{HoJa14}:

\begin{lemma}\label{lem:dbia-almost-equiv-right-invariant}
  Let~$A=(Q,\Sigma,\cdot,\circ,q_0,F)$
  and~$A'=(Q',\Sigma,\cdot',\circ',q_0',F')$ be two biautomata, and
  let~$p\in Q$ and~$q\in Q'$.  Then $p\sim q$ if and only if $(p\cdot
  u)\circ v\sim (q\cdot' u)\circ' v$, for all words~$u,v\in\Sigma^*$.
  Moreover, $p\sim q$ implies~$(p\cdot u)\circ v \equiv (q\cdot'
  u)\circ' v$, for all words~$u,v\in\Sigma^*$ with $\abs{uv}\geq
  k=\abs{Q\times Q'}$.
\end{lemma}

The two automata~$A$ and~$A'$ in
Lemma~\ref{lem:dbia-almost-equiv-right-invariant} need not be
different, so this lemma can also be used for states~$p$ and~$q$ in
one biautomaton~$A$.  Further, if this biautomaton~$A$ is minimal,
then it does not contain a pair of different, but equivalent states.
In this case the states~$p$ and~$q$ are almost-equivalent if and only
if for all long enough words~$uv\in\Sigma^*$ the two states~$(p\cdot
u)\circ v$ and~$(q\cdot u)\circ v$ are the same state.  We obtain the
following corollary.

\begin{corollary}\label{cor:p-q-almost-iff-equal-for-long-words}
  Let $A=(Q,\Sigma,\cdot,\circ,q_0,F)$ be a \emph{minimal} biautomaton,
  and~$k=\abs{Q\times Q}$.  Two states~$p,q\in Q$ are
  almost-equivalent if and only if for all words~$u,v\in\Sigma^*$,
  with $\abs{uv}\geq k$, it is $(p\cdot u)\circ v = (q\cdot u)\circ
  v$. \hfill \qed
\end{corollary}

Now we are ready for our result on strongly permutation-free
biautomata.

\begin{theorem}\label{thm:dbia-characterization-almost-co-finite}
  Let~$A=(Q,\Sigma,\cdot,\circ,q_0,F)$ be some \emph{minimal} biautomaton,
  then the following statements are equivalent:
  \begin{enumerate}
  \item All states in~$Q$ are pairwise
    almost-equivalent.  \label{item:dbia-charact-almost}
  \item There is a state~$s\in Q$ that is a $(u,v)$-attractor in~$A$,
    for all~$u,v\in\Sigma^*$ with $|uv|\geq
    1$. \label{item:dbia-charact-attractor}
  \item $A$ is strongly
    permutation-free.  \label{item:dbia-charact-permutation}
  \item $L(A)$ is a finite or co-finite
    language.\label{item:dbia-charact-definite}
  \end{enumerate}  
\end{theorem}

\begin{proof}
  Let $A=(Q,\Sigma,\cdot,\circ,q_0,F)$ be a minimal biautomaton.
  First assume~$L(A)$ is a finite language, and let~$\ell$ be the
  length of a longest word in~$L(A)$.  Then the right
  language~$L_A(q)$ of every state~$q\in Q$ is finite, so all states
  in~$Q$ are pairwise almost-equivalent.  Since~$L(A)$ is finite,
  and~$A$ is minimal, the biautomaton has a non-accepting sink
  state~$s$, with $s\cdot a = s\circ a = s$ for all~$a\in\Sigma$.
  Moreover, from any state~$q\in Q$ the automaton always reaches this
  sink state~$s$ after reading at most~$\ell$ symbols.  Therefore, the
  state~$s$ is a $(u,v)$-attractor in~$A$, for all
  words~$u,v\in\Sigma^*$ with $|uv|\geq 1$.  It also follows that the
  only permutation that is possible in~$A$ is the identity permutation
  on the singleton set~$\{s\}$, hence~$A$ is strongly
  permutation-free.  The case where~$L(A)$ is a co-finite language is
  similar.  The only differences are that the sink state~$s$ is an
  accepting state, and the integer~$\ell$ must be the length of the
  longest word that is \emph{not} in~$L(A)$.  This shows that
  statement~\ref{item:dbia-charact-definite} implies all other
  statements, and it remains to prove the other directions.

  Assume that all states in~$Q$ are pairwise almost-equivalent, and
  let~$k$ be the integer from
  Corollary~\ref{cor:p-q-almost-iff-equal-for-long-words}.  If the
  length of every word in~$L(A)$ is less than~$k$ then~$L(A)$ is a
  finite language, so assume that there is a word~$w\in L(A)$
  with~$\abs w \geq k$.  We show that in this case language~$L(A)$
  contains every word of length at least~$k$, and thus, is co-finite.
  Since~$\abs w \geq k$, we can write~$w$ as~$w=w_1w_2$,
  with~$\abs{w_2}=k$.  Because~$w\in L(A)$ we have $(q_0\cdot
  w_1)\circ w_2\in F$.  Now let~$u\in\Sigma^{\geq k}$, and consider
  the states~$p=(q_0\cdot w_1)$ and~$q=(q_0\cdot u)$.  Since all
  states are almost-equivalent, and the word~$w_2$ has length~$k$, we
  can use Corollary~\ref{cor:p-q-almost-iff-equal-for-long-words} to
  obtain $(p\cdot \lambda)\circ w_2 = (q\cdot \lambda)\circ w_2$.
  Hence the state $q\circ w_2 = (q_0\cdot u)\circ w_2$ is accepting.
  By the $\diamond$-property of~$A$ we have $(q_0\cdot u)\circ w_2 =
  (q_0\circ w_2)\cdot u$, and another application of
  Corollary~\ref{cor:p-q-almost-iff-equal-for-long-words} on the
  almost-equivalent states~$q_0$ and~$q_0\circ w_2$ we obtain
  $q_0\cdot u = (q_0\circ w_2)\cdot u$, because $\abs u \geq k$.  This
  shows that the word~$u$ is accepted by~$A$, hence~$L(A)$ is
  co-finite.  This shows that
  statements~\ref{item:dbia-charact-almost}
  and~\ref{item:dbia-charact-definite} are equivalent.
  
  Next assume there is a state~$s\in Q$ that is a $(u,v)$-attractor
  for all words~$u,v\in \Sigma^*$ with $|uv|\geq 1$.  Then~$A$ must be
  strongly permutation-free, which can be seen as follows.  Assume
  that there are words~$u,v\in\Sigma^*$ with $|uv|\geq 1$ such that
  the mapping $\pi\colon q\mapsto (q\cdot u)\circ v$ is a permutation
  on some set~$P\subseteq Q$, with~$\abs P \geq 2$.  Then it must be
  $\abs{\pi^i(P)} = \abs P \geq 2$, for all~$i\geq 0$. But since~$s$
  is a $(u,v)$-attractor, there is an integer~$m\geq 0$ such that
  $(q\cdot u^m)\circ v^m = s$, for all states~$q\in Q$.  Then
  $\abs{\pi^m(P)}=1$, which is a contradiction, therefore~$A$ is
  strongly permutation-free.

  Now it is sufficient to show that
  statement~\ref{item:dbia-charact-permutation} implies
  statement~\ref{item:dbia-charact-almost}.  Therefore let~$A$ be
  strongly permutation-free, and assume for the sake of contradiction,
  that there are two states~$p,q\in Q$ with~$p\nsim q$.
  Corollary~\ref{cor:p-q-almost-iff-equal-for-long-words} now implies
  that there are words~$u,v\in\Sigma^*$, with~$\abs{uv}\geq
  \abs{Q\times Q}$, such that $(p\cdot u)\circ v \neq (q\cdot u)\circ
  v$.
  Since the number of steps in the computations $(p\cdot u)\circ v$
  and $(q\cdot u)\circ v$ is at least~$\abs{Q\times Q}$, the words~$u$
  and~$v$ can be written as~$u=u_1u_2u_3$ and~$v=v_3v_2v_1$ such that
  \begin{align*}
    (p\cdot u_1)\circ v_1 &= p_1, & (p_1\cdot u_2)\circ v_2 &= p_1, &
    (p_1\cdot u_3)\circ v_3 &= (p\cdot u)\circ v,\\
    (q\cdot u_1)\circ v_1 &= q_1, & (q_1\cdot u_2)\circ v_2 &= q_1, &
    (q_1\cdot u_3)\circ v_3 &= (q\cdot u)\circ v.
  \end{align*}
  But then the mapping $\pi\colon r\mapsto (r\cdot u_2)\circ v_2$ is a
  permutation (the identity permutation) on the states~$\{p_1,q_1\}$.
  Since~$A$ is strongly permutation-free, it follows~$p_1=q_1$, and in
  turn $(p\cdot u)\circ v = (q\cdot u)\circ v$.  This contradicts the
  assumption~$p\nsim q$, and concludes our proof.
\end{proof}

\section{Ordered Automata}
\label{sec:ordered-automata}

We now study automata where one can find an order on the state set
that is compatible with the transitions of the automaton.  Ordered
\dfa s and their accepted languages were studied in~\cite{ShTh74}.  A
\dfa\ $A=(Q,\Sigma,\delta,q_0,F)$ is \emph{ordered} if there exists
some total order~$\leq$ on the state set~$Q$ such that~$p\leq q$
implies $\delta(p,a)\leq \delta(q,a)$, for all states~$p,q\in Q$ and
symbols~$a\in\Sigma$.  Similarly, a biautomaton
$A=(Q,\Sigma,\cdot,\circ,q_0,F)$ is \emph{ordered} if there is a total
order~$\leq$ on~$Q$ such that~$p\leq q$ implies $p\cdot a \leq q\cdot
a$ as well as $p\circ a \leq q\circ a$, for all states~$p,q\in Q$ and
symbols~$a\in\Sigma$.
A regular language is \emph{ordered} if it is accepted by some ordered
\dfa, and it is \emph{bi-ordered} if it is accepted by an ordered
biautomaton.  Moreover, a language is \emph{strictly ordered} if its
minimal \dfa\ is ordered, and it is \emph{strictly bi-ordered} if its
minimal biautomaton is ordered.

The next two results show that the class of bi-ordered languages is
located between the class of ordered languages and the class of finite
and co-finite languages.

\begin{theorem}\label{thm:bi-ordered-strictly-in-ordered}
  The class of bi-ordered languages is strictly contained in the class
  of ordered languages.  
\end{theorem}

\begin{proof}
  If~$L$ is a bi-ordered language, then it is accepted by some ordered
  \dbia~$A$.  Then of course the automaton~$A_\fwd$ is an ordered
  \dfa.  Therefore any bi-ordered language is an ordered language.
  The strictness of this inclusion is witnessed by the
  language~$\Sigma^* ab \Sigma^*$ over the alphabet~$\Sigma=\{a,b\}$,
  which is accepted by the ordered \dfa~$A$ from
  Figure~\ref{fig:ordered-dfa-SigStar-ab-SigStar}.
  \begin{figure}
    \centering
    \includegraphics[scale=.9]{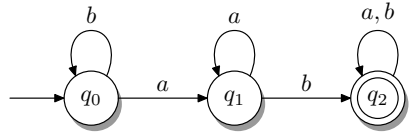}
    \caption{The minimal \dfa~$A$ for the language $\Sigma^* ab
      \Sigma^*$ over the alphabet~$\Sigma=\{a,b\}$.  The states
      of~$A_1$ can be ordered by $q_0 \leq q_1 \leq q_2$.}
    \label{fig:ordered-dfa-SigStar-ab-SigStar}
  \end{figure}
  
  Next let us argue, why no biautomaton for the language $\Sigma^* ab
  \Sigma^*$ can be ordered.  Therefore consider some biautomaton
  $B=(Q,\Sigma,\cdot,\circ,q_0,F)$ with $L(B)=\Sigma^* ab\Sigma^*$.
  In the following we use the notation~$[u.v]$, for~$u,v\in\Sigma^*$,
  to describe the state~$(q_0\cdot u)\circ v$ of~$B$.  Of course,
  different word pairs~$[u.v]$ and~$[u'.v']$ may describe the same
  state.  First note that the three states~$[\lambda.\lambda]$,
  $[a.\lambda]$, and $[\lambda.b]$ must be pairwise distinct, because
  every one of these states leads to an accepting state on a different
  input string.  Assume there is some order~$\leq$ on the state
  set~$Q$ that is compatible with the transition functions of~$B$.
  There are six different possibilities to order the three above
  mentioned states:
  \begin{align*}
    [\lambda.\lambda] &\leq [a.\lambda] \leq [\lambda.b], 
    & [a.\lambda] &\leq [\lambda.\lambda] \leq [\lambda.b],
    & [\lambda.b] &\leq [a.\lambda] \leq [\lambda.\lambda], \\
    [\lambda.\lambda] &\leq [\lambda.b] \leq [a.\lambda],
    & [a.\lambda] &\leq [\lambda.b] \leq [\lambda.\lambda],
    & [\lambda.b] &\leq [\lambda.\lambda] \leq [a.\lambda].
  \end{align*}
  If $[\lambda.\lambda] \leq [a.\lambda] \leq [\lambda.b]$ then it
  must be $[\lambda.b^i] \leq [a.b^i] \leq [\lambda.b^{i+1}]$, for
  all~$i\geq 0$.  Since the number of states in~$Q$ is finite, there
  must be integers~$j>k\geq 0$ for which $[\lambda.b^k] =
  [\lambda.b^j]$.  It then follows that $[\lambda.b^k] = [a.b^k] =
  [\lambda.b^{k+1}]$.  This is a contradiction because~$[a.b^k]$
  describes an accepting state, while~$[\lambda.b^k]$
  and~$[\lambda.b^{k+1}]$ describe non-accepting states.

  If $[\lambda.\lambda] \leq [\lambda.b] \leq [a.\lambda]$ then we
  obtain $[a^i.\lambda] \leq [a^i.b] \leq [a^{i+1}.\lambda]$ for
  all~$i\geq 0$.  Similar to the case above we get a contradiction:
  because~$Q$ is finite there is an integer~$k\geq 0$ with
  $[a^k.\lambda] = [a^k.b] = [a^{k+1}.\lambda]$, but the
  state~$[a^k.b]$ is accepting while the other two states are
  non-accepting.

  Next consider the case $[a.\lambda] \leq [\lambda.\lambda] \leq
  [\lambda.b]$.  By reading symbol~$a$ with a forward transition we
  obtain $[a.\lambda]\leq [a.b]$ from the second inequality, and by
  reading~$b$ with a backward transition, the first inequality implies
  $[a.b]\leq [\lambda.b]$.  Note that both times~$[a.b]$ describes the
  same state because~$B$ has the $\diamond$-property.  Further, this
  state must be different from the three states~$[a.\lambda]$,
  $[\lambda.\lambda]$, and~$[\lambda.b]$ because it is an accepting
  state, and the others are not.  Now there are two possibilities for
  the placement of state~$[a.b]$ in the order:
  \[ [a.\lambda] \leq [a.b]\leq [\lambda.\lambda] \leq [\lambda.b]
  \quad{\text{or}}\quad [a.\lambda] \leq [\lambda.\lambda] \leq [a.b]
  \leq [\lambda.b]. \]
  The first case implies $[a^{i+1}.\lambda]\leq [a^{i+1}.b] \leq
  [a^i.\lambda]$, for all~$i\geq 0$, which leads to the contradictory
  equation $[a^{k+1}.\lambda] = [a^{k+1}.b] = [a^k.\lambda]$, for
  some~$k\geq 0$.  The second case implies $[\lambda.b^i] \leq
  [a.b^{i+1}] \leq [\lambda.b^{i+1}]$, for all~$i\geq 0$, and to the
  contradiction $[\lambda.b^k] = [a.b^{k+1}] = [\lambda.b^{k+1}]$, for
  some~$k\geq 0$.

  With similar argumentation, the remaining three cases $[a.\lambda]
  \leq [\lambda.b] \leq [\lambda.\lambda]$, $[\lambda.b] \leq
  [a.\lambda] \leq [\lambda.\lambda]$, and $[\lambda.b] \leq
  [\lambda.\lambda] \leq [a.\lambda]$ lead to contradictions---we omit
  the details.  This shows that there is no order of the state set~$Q$
  that is compatible with the transition functions of~$B$.  Therefore,
  the ordered language~$\Sigma^* ab\Sigma^*$ is not a bi-ordered
  language.
\end{proof}

In the proof of the following result we use the lexicographic
order~$<_{\text{lex}}$ of words, which is defined as follows.
Let~$\Sigma$ be an alphabet of size~$k$ and fix some
order~$a_1,a_2,\dots,a_k$ of the symbols from~$\Sigma$.  For two
words~$w_1,w_2\in\Sigma^*$ let~$w_1 <_{\text{lex}} w_2$ if and only if
either~$w_1$ is a prefix of~$w_2$, or $w_1= u a_i w_1'$ and $w_2= u
a_j w_2'$, for some words~$u,w_1',w_2'\in\Sigma^*$ and
symbols~$a_i,a_j\in\Sigma$, with~$i<j$.

\begin{theorem}\label{thm:finite-strictly-in-bi-ordered}
  The class of finite and co-finite languages is strictly contained in
  the class of bi-ordered languages.
\end{theorem}

\begin{proof}
  Let~$L$ be some finite language over the alphabet~$\Sigma$ and
  let~$\ell$ be the length of the longest word in~$L$.  We construct
  an ordered biautomaton for~$L$ as follows.  Let
  $A=(Q,\Sigma,\cdot,\circ,q_0,F)$ be the biautomaton with state set
  $Q=\{\,(u,v)\mid u,v\in\Sigma^*, \abs{uv}\leq \ell\,\} \cup\{ s \}$,
  initial state~$q_0 = (\lambda,\lambda)$, set of accepting states
  $F=\{\,(u,v)\in Q\mid uv\in L\,\}$,
  and where the transition functions~$\cdot$ and~$\circ$ are defined
  as follows: for all symbols~$a\in\Sigma$ let $s\cdot a = s\circ a =
  s$, and for all states~$(u,v)\in Q$ let
  \begin{align*}
    (u,v)\cdot a &=
    \begin{cases}
      (ua,v) &\text{if $\abs{uav}\leq\ell$,}\\
      s &\text{otherwise,}
    \end{cases}
    &    
    (u,v)\circ a &=
    \begin{cases}
      (u,av) &\text{if $\abs{uav}\leq\ell$,}\\
      s &\text{otherwise.}
    \end{cases}
  \end{align*}
  One readily sees that~$A$ has both the $\diamond$-property, and the
  $F$-property.  Now let us define the order~$\leq$ on~$Q$ as follows.
  First of all let~$(u,v)\leq s$, for all~$(u,v)\in Q$, so the
  non-accepting sink state is the largest element of~$Q$.  Next, for
  two different states~$(u_1,v_1)$ and~$(u_2,v_2)$ let $(u_1,v_1)\leq
  (u_2,v_2)$ if and only if
  \begin{itemize}
  \item $\abs{u_1v_1} < \abs{u_2v_2}$, or
  \item $\abs{u_1v_1} = \abs{u_2v_2}$, and $\abs{u_1} < \abs{u_2}$, or
  \item $\abs{u_1v_1} = \abs{u_2v_2}$, $\abs{u_1} = \abs{u_2}$, and
    $u_1 <_{\text{lex}} u_2$, or
  \item $\abs{u_1v_1} = \abs{u_2v_2}$, $u_1 = u_2$, and $v_1
    <_{\text{lex}} v_2$.
  \end{itemize}
  Notice that if none of the four cases above holds, then
  $(u_1,v_1)=(u_2,v_2)$.  It remains to show that the transitions
  of~$A$ respect the order~$\leq$.  Since state~$s$ goes to itself on
  every symbol, because it is the largest element, we have $(u,v)\cdot
  a \leq s\cdot a$, and $(u,v)\circ a \leq s\circ a$.  Next
  let~$(u_1,v_1)$ and~$(u_2,v_2)$ be two different states of~$A$ with
  $(u_1,v_1)\leq (u_2,v_2)$.  Then it must be~$\abs{u_1v_1}\leq
  \abs{u_2v_2}\leq \ell$.  If~$\abs{u_2v_2} = \ell$, then~$(u_2,v_2)$
  goes to the sink state~$s$ on both the forward, and the backward
  $a$-transition.  Since~$s$ is the largest element we obtain
  $(u_1,v_1)\cdot a \leq (u_2,v_2)\cdot a$, and $(u_1,v_1)\circ a \leq
  (u_2,v_2)\circ a$.  Therefore, in the following argumentation we
  assume that $\abs{u_2,v_2}<\ell$, so that we have $(u_1,v_1)\cdot a
  = (u_1a,v_1)$ and $(u_2,v_2)\cdot a = (u_2a,v_2)$, as well as
  $(u_1,v_1)\circ a = (u_1,av_1)$ and $(u_2,v_2)\circ a = (u_2,av_2)$.
  Now we have to show $(u_1a,v_1)\leq (u_2a,v_2)$ and $(u_1,av_1)\leq
  (u_2,av_2)$, for which we distinguish four cases.
  \begin{itemize}
  \item If $\abs{u_1v_1} < \abs{u_2v_2}$ then clearly $\abs{u_1av_1} <
    \abs{u_2av_2}$, from which we conclude $(u_1a,v_1) \leq
    (u_2a,v_2)$, and $(u_1,av_1) \leq (u_2,av_2)$.
  \item If $\abs{u_1v_1} = \abs{u_2v_2}$, and $\abs{u_1} < \abs{u_2}$,
    then also $\abs{u_1a} < \abs{u_2a}$.  Again, we can conclude
    $(u_1a,v_1)\leq (u_2a,v_2)$, and $(u_1,av_1)\leq (u_2,av_2)$.
  \item Next assume $\abs{u_1v_1} = \abs{u_2v_2}$, $\abs{u_1} =
    \abs{u_2}$, and $u_1 <_{\text{lex}} u_2$.  Since~$u_1$ and~$u_2$
    are of same length, the fact~$u_1 <_{\text{lex}} u_2$ implies
    $u_1a <_{\text{lex}} u_2a$.  Thus, we get $(u_1a,v_1)\leq
    (u_2a,v_2)$, and $(u_1,av_1)\leq (u_2,av_2)$.
  \item Finally let $\abs{u_1v_1} = \abs{u_2v_2}$, $\abs{u_1} =
    \abs{u_2}$, $u_1 = u_2$, and $v_1 <_{\text{lex}} v_2$.  From $v_1
    <_{\text{lex}} v_2$ follows $av_1 <_{\text{lex}} av_2$, so we
    obtain $(u_1a,v_1)\leq (u_2a,v_2)$, and $(u_1,av_1)\leq
    (u_2,av_2)$.
  \end{itemize}
  This shows that the biautomaton~$A$ is an ordered biautomaton.

  In case of a co-finite language~$L\subseteq\Sigma^*$ we first take
  its complement $\Sigma^*\setminus L$, which is finite, and apply the
  above given construction. Then we obtain an ordered
  biautomata~$A$. Finally, exchanging accepting and non-accepting
  states---this is the ordinary complementation construction known for
  \dfa s applied to \dbia s---results in an ordered biautomata for the
  language~$L$.

  Finally, strictness of the inclusion is witnessed by the infinite
  and not co-finite language~$a^*+b$, which is accepted by the
  bi-ordered biautomaton from
  Figure~\ref{fig:ordered-biaut-for-a-star-plus-b}.
  \begin{figure}[b]
    \centering
    \includegraphics[scale=.9]{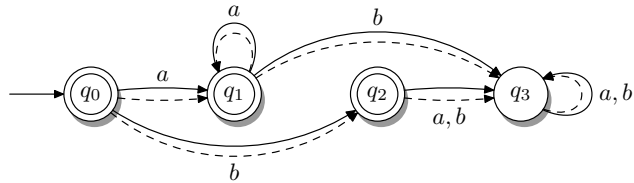}
    \caption{A bi-ordered biautomaton with order~$q_0\leq q_1\leq q_2\leq q_3$
      for the language~$a^*+b$.}
    \label{fig:ordered-biaut-for-a-star-plus-b}
  \end{figure}
\end{proof}

Notice that the language~$\Sigma^* ab \Sigma^*$ from the proof of
Theorem~\ref{thm:bi-ordered-strictly-in-ordered} is even a strictly
ordered language, since its minimal \dfa~$A$ from
Figure~\ref{fig:ordered-dfa-SigStar-ab-SigStar} is ordered.  As we
have seen, this language is not a bi-ordered language, therefore the
class of bi-ordered languages does not even contain all strictly
ordered languages.  On the other hand, if a language is strictly
bi-ordered, i.e., if its minimal biautomaton~$B$ is ordered, then also
the minimal \dfa~$B_\fwd$ is ordered.  Therefore, the class of
strictly bi-ordered languages is contained in the classes of strictly
ordered languages.  We summarize our findings in the following
corollary.

\begin{corollary}\label{cor:strictly-bi-ordered-strictly-in-strictly-ordered}
  The class of strictly bi-ordered languages is proper subset of 
  the class of strictly ordered languages.  \hfill \qed
\end{corollary}

Concerning the relation between bi-ordered languages and
strictly ordered languages, we can see that these are incomparable to
each other.  We have seen that the strictly ordered language~$\Sigma^*
ab \Sigma^*$ is not bi-ordered.  On the other hand, we know that every
finite language is bi-ordered.  But one can see that the minimal \dfa\
for the finite language~$\{ab\}$ is not ordered---the reader is
invited convince himself of this fact.  A proof of a more general
result, saying that a single word language is strictly ordered if and
only if the word is of the form~$a^i$ for some alphabet symbol~$a$ and
integer~$i\geq 0$, can be found in~\cite{ShTh74}.

\begin{corollary}
  The classes of bi-ordered languages and of strictly ordered
  languages are incomparable to each other.  \hfill \qed
\end{corollary}

Moreover, with a similar argumentation as above we obtain:

\begin{corollary}
  The class strictly bi-ordered languages is a proper subset of the
  class of bi-ordered languages.  \hfill \qed
\end{corollary}

\section{Non-Exiting and Non-Returning Automata}
\label{sec:non-exiting-non-returning}

In this last section we study so called non-exiting automata and
non-returning automata.  
A biautomaton or finite automaton~$A$ is \emph{non-exiting} if all
outgoing transitions from accepting states go to a non-accepting sink
state, and it is \emph{non-returning} if the initial state does not
have any ingoing transitions.  We say that~$A$ is \emph{exiting} if it
is \emph{not} non-exiting, and it is \emph{returning} if it is
\emph{not} non-returning.

In the \dfa\ case non-exiting automata are known to characterize the
class of prefix-free languages, while non-returning automata are
related to suffix-free languages.  However this latter relation is not
a characterization: it is true that every \dfa\ that accepts a
(non-empty) suffix-free language must be non-returning, but the
reverse implication does not hold.

Concerning the situation for biautomata, we first show that every
biautomaton which is non-exiting must also be non-returning (unless it
accepts the empty language).

\begin{lemma}\label{lem:non-exiting-implies-non-returning}
  Let~$A$ be a biautomaton with~$L(A)\neq\emptyset$.  If~$A$ is
  non-exiting, then~$A$ is non-returning.
\end{lemma}

\begin{proof}
  We prove the contraposition of the lemma.
  Assume~$A=(Q,\Sigma,\cdot,\circ,q_0,F)$ is a returning biautomaton.
  Then there must be words~$u,v\in\Sigma^*$, with~$\abs{uv}\geq 1$,
  such that $(q_0\cdot u)\circ v = q_0$.  It follows that $(q_0\cdot
  u^n)\circ v^n = (q_0\circ v^n)\cdot u^n = q_0$, for all~$n\geq 0$.
  Since the number of states in~$A$ is finite, there are
  integers~$i,j\geq 0$ and~$x,y\geq 1$ such that $q_0\cdot u^i =
  q_0\cdot u^{i+x}$ and $q_0\circ v^j = q_0\circ v^{j+y}$.  We obtain
  $q_0 = (q_0\cdot u^{i+x})\circ v^i$ and $q_0 = (q_0\circ
  v^{j+y})\cdot u^j$.
  Now let~$w\in L(A)$, i.e., $q_0\cdot w\in F$.  From our
  considerations above we get $((q_0\cdot u^i)\circ v^i)\cdot w \in F$
  and
  \[ ((q_0\cdot u^i)\circ v^i)\cdot w=((q_0\cdot u^{i+x})\circ v^i)\cdot w = (((q_0\cdot u^i)\circ
  v^i)\cdot w)\cdot u^x \in F. \]
  Recall that~$\abs{uv}\geq 1$.  If~$\abs u \geq 1$, then we see
  that~$A$ is exiting because the accepting state $((q_0\cdot
  u^{i+x})\circ v^i)\cdot w$ cannot go to a non-accepting sink state
  on every input symbol.  If~$\abs u = 0$ then it must be~$\abs v\geq
  1$.  Now a similar argumentation gives $((q_0\circ v^j)\cdot
  u^j)\cdot w \in F$ and
  \[ ((q_0\circ v^j)\cdot
  u^j)\cdot w=((q_0\circ v^{j+y})\cdot u^j)\cdot w = (((q_0\circ v^j)\cdot
  u^j)\cdot w)\circ v^y \in F. \]
  Here the accepting state $((q_0\circ v^j)\cdot u^j)\cdot w$ cannot
  go to a non-accepting sink state on every alphabet symbol, which
  shows that~$A$ is exiting.
\end{proof}

The converse of Lemma~\ref{lem:non-exiting-implies-non-returning} is
not true which can easily be seen by the minimal biautomaton for the
language~$\{a,aa\}$.  Since the language is finite, there cannot be a
cycle $q_0 = (q_0\cdot u)\circ v$, hence the biautomaton is
non-returning.  However, since both states $q_0\cdot a$ and $(q_0\cdot
a) \cdot a$ are accepting, the automaton cannot be non-exiting.

Now we study the classes of languages accepted by biautomaton that are
non-exiting or non-returning.
While minimal non-exiting \dfa s characterize the class of prefix-free
languages, we show in the following that minimal non-exiting
biautomata characterize a different language class, namely the class
of circumfix-free languages.  A word~$v\in\Sigma^*$ is a
\emph{circumfix} of a word~$w\in\Sigma^*$, if~$w=w_1w_2w_3$
and~$v=w_1w_3$, for some words~$w_1,w_2,w_3\in\Sigma^*$.  A
language~$L$ is called \emph{circumfix-free} if there are no two
\emph{different} words~$w,v\in L$, such that~$v$ is a circumfix of~$w$.
Notice that prefixes and suffixes of a word are also circumfixes
(where one ``side'' is $\lambda$).  Therefore the class of
circumfix-free languages is contained in both the classes of
prefix-free languages and suffix-free languages.

\begin{theorem}\label{thm:dbia-non-exiting-iff-circumfix-free}
  A regular language is circumfix-free if and only if its \emph{minimal}
  biautomaton is non-exiting.
\end{theorem}

\begin{proof}
  Let $A=(Q,\Sigma,\cdot,\circ,q_0,F)$ be a minimal biautomaton
  and~$L=L(A)$.  First assume that~$A$ is exiting, i.e., there is an
  accepting state~$q\in F$ and a non-empty word~$w\in\Sigma^+$ such
  that~$q\cdot w\in F$.  Since~$q$ must be reachable from the initial
  state of~$A$, there are words~$u,v\in\Sigma^*$ with~$(q_0\cdot
  u)\circ v= q$.  Then both words~$uv$ and~$uwv$ belong to~$L(A)$,
  but~$uv$ is a circumfix of~$w$.  Therefore, if~$L(A)$ is
  circumfix-free then~$A$ must be non-exiting.

  For the reverse implication notice that whenever there are two
  different words~$w$ and~$w'$ in~$L(A)$ such that~$w'$ is a circumfix
  of~$w$, then~$w=w_1w_2w_3$ and~$w'=w_1w_3$,
  with~$w_1,w_3\in\Sigma^*$ and~$w_2\in\Sigma^+$ (because~$w\neq w'$).
  It follows $(q_0\cdot w_1)\circ w_3\in F$ and $((q_0\cdot w_1)\circ
  w_3)\cdot w_2\in F$, and since~$w_2\neq\lambda$ the automaton~$A$
  must be exiting.  Thus, if~$A$ is non-exiting then~$L(A)$ must be
  circumfix-free.
\end{proof}

Now we consider languages accepted non-returning biautomata.  If a
minimal biautomaton~$A$ is non-returning then clearly the contained
minimal \dfa~$A_\fwd$ is non-returning, too.  Therefore the class of
languages accepted by minimal non-returning biautomata is contained in
the class of languages accepted by minimal non-returning \dfa s.
Moreover, this inclusion is strict because the minimal \dfa\ for the
language~$ab^*$ is non-returning, while the minimal biautomaton for
that language is not (it has a backward transition loop for symbol~$b$
on its initial state).  Therefore we have the following result.

\begin{theorem}
  The class of languages accepted by \emph{minimal} non-returning biautomata
  is strictly contained in the class of languages accepted by \emph{minimal}
  non-returning deterministic finite automata. \hfill \qed
\end{theorem}

\section{Conclusions}
\label{sec:conclusions}

We continued the study of structural properties on biautomata started
in~\cite{HoJa13c,KlPo12a,KlPo12}. Our focus was on the effect of
classical properties of deterministic finite automata such as, e.g.,
permutation-freeness, strongly permutation-freeness, and orderability,
on biautomata. It is shown that this approach on structurally
restricting the recently introduced biautomata model was worth looking
at. A comparison of the induced language families on structurally
restricted deterministic automata and biautomata is given in
Table~\ref{tab:results}. Future research on the subject under
consideration may consist on some further properties such as, e.g.,
biautomata where all states are final or all are initial. In the case
of ordinary deterministic finite automata the family of prefix-closed
languages is obtained by the former property, while the latter gives
the family of suffix-closed languages. Moreover, it would be also
interesting to study, which structural properties can be successfully
applied to nondeterministic biautomata, as introduced
in~\cite{HoJa13a}.

\def\lastmodified{23.12.2008}\def\id#1{#1}


\begin{thebibliography}{10}
\providecommand{\bibitemdeclare}[2]{}
\providecommand{\surnamestart}{}
\providecommand{\surnameend}{}
\providecommand{\urlprefix}{Available at }
\providecommand{\url}[1]{\texttt{#1}}
\providecommand{\href}[2]{\texttt{#2}}
\providecommand{\urlalt}[2]{\href{#1}{#2}}
\providecommand{\doi}[1]{doi:\urlalt{http://dx.doi.org/#1}{#1}}
\providecommand{\bibinfo}[2]{#2}

\bibitemdeclare{book}{Ar69}
\bibitem{Ar69}
\bibinfo{author}{M.~A. \surnamestart Arbib\surnameend} (\bibinfo{year}{1969}):
  \emph{\bibinfo{title}{Theories of Abstract Automata}}.
\newblock \bibinfo{series}{Automatic Computation},
  \bibinfo{publisher}{Prentice-Hall}, \bibinfo{address}{London}.

\bibitemdeclare{article}{BGS09}
\bibitem{BGS09}
\bibinfo{author}{A.~\surnamestart Badr\surnameend},
  \bibinfo{author}{V.~\surnamestart Geffert\surnameend} \&
  \bibinfo{author}{I.~\surnamestart Shipman\surnameend} (\bibinfo{year}{2009}):
  \emph{\bibinfo{title}{Hyper-Minimizing Minimized Deterministic Finite State
  Automata}}.
\newblock {\sl \bibinfo{journal}{RAIRO--Informatique th{\'e}orique et
  Applications / Theoretical Informatics and Applications}}
  \bibinfo{volume}{43}(\bibinfo{number}{1}), pp. \bibinfo{pages}{69--94}.
\doi{10.1051/ita:2007061}

\bibitemdeclare{misc}{BrLi12}
\bibitem{BrLi12}
\bibinfo{author}{J.~\surnamestart Brzozowski\surnameend} \&
  \bibinfo{author}{B.~\surnamestart Liu\surnameend} (\bibinfo{year}{2021}):
  \emph{\bibinfo{title}{Syntactic Complexity of Finite/Cofinite, Definite, and
  Reverse Definite Languages}}.
\newblock \bibinfo{howpublished}{arXiv:1203.2873v1 [cs.FL]}.

\bibitemdeclare{article}{BrFi80}
\bibitem{BrFi80}
\bibinfo{author}{J.~A. \surnamestart Brzozowski\surnameend} \&
  \bibinfo{author}{F.~E. \surnamestart Fitch\surnameend}
  (\bibinfo{year}{1980}): \emph{\bibinfo{title}{Languages of
  $\mathcal{R}$-Trivial Monoids}}.
\newblock {\sl \bibinfo{journal}{Journal of Computer and System Sciences}}
  \bibinfo{volume}{20}(\bibinfo{number}{1}), pp. \bibinfo{pages}{32--49}.
\doi{10.1016/0022-0000(80)90003-3}

\bibitemdeclare{inproceedings}{CDJM13a}
\bibitem{CDJM13a}
\bibinfo{author}{J.-M. \surnamestart Champarnaud\surnameend},
  \bibinfo{author}{J.-P. \surnamestart Dubernard\surnameend},
  \bibinfo{author}{H.~\surnamestart Jeanne\surnameend} \&
  \bibinfo{author}{L.~\surnamestart Mignot\surnameend} (\bibinfo{year}{2013}):
  \emph{\bibinfo{title}{Two-Sided Derivatives for Regular Expressions and for
  Hairpin Expressions}}.
\newblock In \bibinfo{editor}{A.~H. \surnamestart Dediu\surnameend},
  \bibinfo{editor}{C.~\surnamestart Mart{\'\i}n-Vide\surnameend} \&
  \bibinfo{editor}{B.~\surnamestart Truthe\surnameend}, editors: {\sl
  \bibinfo{booktitle}{Proc.\ of the $7$th International Conference on
  Language and Automata Theory and Applications}}, {\sl \bibinfo{series}{LNCS}}
  \bibinfo{volume}{7810}, \bibinfo{publisher}{Springer},
  \bibinfo{address}{Bilbao, Spain}, pp. \bibinfo{pages}{202--213}.
\doi{10.1007/978-3-642-37064-9\_19}

\bibitemdeclare{article}{Ha69a}
\bibitem{Ha69a}
\bibinfo{author}{I.~M. \surnamestart Havel\surnameend} (\bibinfo{year}{1969}):
  \emph{\bibinfo{title}{The theory of regular events {II}}}.
\newblock {\sl \bibinfo{journal}{Kybernetika}} \bibinfo{volume}{6}, pp.
  \bibinfo{pages}{520--544}.

\bibitemdeclare{inproceedings}{HoJa13c}
\bibitem{HoJa13c}
\bibinfo{author}{M.~\surnamestart Holzer\surnameend} \&
  \bibinfo{author}{S.~\surnamestart Jakobi\surnameend} (\bibinfo{year}{2013}):
  \emph{\bibinfo{title}{Minimization and Characterizations for Biautomata}}.
\newblock In \bibinfo{editor}{S.~\surnamestart Bensch\surnameend},
  \bibinfo{editor}{F.~\surnamestart Drewes\surnameend},
  \bibinfo{editor}{R.~\surnamestart Freund\surnameend} \&
  \bibinfo{editor}{F.~\surnamestart Otto\surnameend}, editors: {\sl
  \bibinfo{booktitle}{Proc.\ of the $5$th International Workshop on
  Non-Classical Models of Automata and Applications}}, {\sl
  \bibinfo{series}{books$@$ocg.at}} \bibinfo{volume}{294},
  \bibinfo{publisher}{\"Osterreichische Computer Gesellschaft},
  \bibinfo{address}{Ume\aa, Sweden}, pp. \bibinfo{pages}{179--193}.

\bibitemdeclare{inproceedings}{HoJa13a}
\bibitem{HoJa13a}
\bibinfo{author}{M.~\surnamestart Holzer\surnameend} \&
  \bibinfo{author}{S.~\surnamestart Jakobi\surnameend} (\bibinfo{year}{2013}):
  \emph{\bibinfo{title}{Nondeterministic Biautomata and Their Descriptional
  Complexity}}.
\newblock In \bibinfo{editor}{H.~\surnamestart J{\"u}rgensen\surnameend} \&
  \bibinfo{editor}{R.~\surnamestart Reis\surnameend}, editors: {\sl
  \bibinfo{booktitle}{Proc.\ of the $15$th International Workshop on
  Descriptional Complexity of Formal Systems}}, {\sl \bibinfo{series}{LNCS}}
  \bibinfo{volume}{8031}, \bibinfo{publisher}{Springer},
  \bibinfo{address}{London, Ontario, Canada}, pp. \bibinfo{pages}{112--123}.
\doi{10.1007/978-3-642-39310-5\_12}

\bibitemdeclare{techreport}{HoJa14}
\bibitem{HoJa14}
\bibinfo{author}{M.~\surnamestart Holzer\surnameend} \&
  \bibinfo{author}{S.~\surnamestart Jakobi\surnameend} (\bibinfo{year}{2014}):
  \emph{\bibinfo{title}{Minimal and Hyper-Minimal Biautomata}}.
\newblock \bibinfo{type}{IFIG Research Report} \bibinfo{number}{1401},
  \bibinfo{institution}{Institut f\"ur Informatik},
  \bibinfo{address}{Justus-Liebig-Universit{\"a}t Gie\ss en, Arndtstr.~2,
  D-35392 Gie\ss en, Germany}.

\bibitemdeclare{inproceedings}{JiKl12}
\bibitem{JiKl12}
\bibinfo{author}{G.~\surnamestart Jir{\'a}skov{\'a}\surnameend} \&
  \bibinfo{author}{O.~\surnamestart Kl{\'\i}ma\surnameend}
  (\bibinfo{year}{2012}): \emph{\bibinfo{title}{Descriptional Complexity of
  Biautomata}}.
\newblock In \bibinfo{editor}{M.~\surnamestart Kutrib\surnameend},
  \bibinfo{editor}{N.~\surnamestart Moreira\surnameend} \&
  \bibinfo{editor}{R.~\surnamestart Reis\surnameend}, editors: {\sl
  \bibinfo{booktitle}{Proc.\ of the $14$th International Workshop
  Descriptional Complexity of Formal Systems}}, {\sl \bibinfo{series}{LNCS}}
  \bibinfo{volume}{7386}, \bibinfo{publisher}{Springer},
  \bibinfo{address}{Braga, Portugal}, pp. \bibinfo{pages}{196--208}.
\doi{10.1007/978-3-642-31623-4\_15}

\bibitemdeclare{inproceedings}{KlPo12a}
\bibitem{KlPo12a}
\bibinfo{author}{O.~\surnamestart Kl{\'\i}ma\surnameend} \&
  \bibinfo{author}{L.~\surnamestart Pol{\'a}k\surnameend}
  (\bibinfo{year}{2012}): \emph{\bibinfo{title}{Biautomata for $k$-Piecewise
  Testable Languages}}.
\newblock In \bibinfo{editor}{H.-C. \surnamestart Yen\surnameend} \&
  \bibinfo{editor}{O.~H. \surnamestart Ibarra\surnameend}, editors: {\sl
  \bibinfo{booktitle}{Proc.\ of the $16$th International Conference
  Developments in Language Theory}}, {\sl \bibinfo{series}{LNCS}}
  \bibinfo{volume}{7410}, \bibinfo{publisher}{Springer},
  \bibinfo{address}{Taipei, Taiwan}, pp. \bibinfo{pages}{344--355}.
\doi{10.1007/978-3-642-31653-1\_31}

\bibitemdeclare{article}{KlPo12}
\bibitem{KlPo12}
\bibinfo{author}{O.~\surnamestart Kl{\'\i}ma\surnameend} \&
  \bibinfo{author}{L.~\surnamestart Pol{\'a}k\surnameend}
  (\bibinfo{year}{2012}): \emph{\bibinfo{title}{On Biautomata}}.
\newblock {\sl \bibinfo{journal}{RAIRO--Informatique th{\'e}orique et
  Applications / Theoretical Informatics and Applications}}
  \bibinfo{volume}{46}(\bibinfo{number}{4}), pp. \bibinfo{pages}{573--592}.
\doi{10.1051/ita/2012014}

\bibitemdeclare{inproceedings}{Lo07}
\bibitem{Lo07}
\bibinfo{author}{R.~\surnamestart Loukanova\surnameend} (\bibinfo{year}{2007}):
  \emph{\bibinfo{title}{Linear Context Free Languages}}.
\newblock In \bibinfo{editor}{C.~B. \surnamestart Jones\surnameend},
  \bibinfo{editor}{Z.~\surnamestart Liu\surnameend} \&
  \bibinfo{editor}{J.~\surnamestart Woodcock\surnameend}, editors: {\sl
  \bibinfo{booktitle}{Proc.\ of the $4$th International Colloquium
  Theoretical Aspects of Computing}}, {\sl \bibinfo{series}{LNCS}}
  \bibinfo{volume}{4711}, \bibinfo{publisher}{Springer},
  \bibinfo{address}{Macau, China}, pp. \bibinfo{pages}{351--365}.

\bibitemdeclare{book}{McNaPa71}
\bibitem{McNaPa71}
\bibinfo{author}{R.~\surnamestart McNaughton\surnameend} \&
  \bibinfo{author}{S.~\surnamestart Papert\surnameend} (\bibinfo{year}{1971}):
  \emph{\bibinfo{title}{Counter-free automata}}.
\newblock {\sl \bibinfo{series}{Research monographs}}~\bibinfo{volume}{65},
  \bibinfo{publisher}{MIT Press}.

\bibitemdeclare{book}{Mi67}
\bibitem{Mi67}
\bibinfo{author}{M.~L. \surnamestart Minsky\surnameend} (\bibinfo{year}{1967}):
  \emph{\bibinfo{title}{Computation: Finite and Infinite Machines}}.
\newblock \bibinfo{series}{Automatic Computation},
  \bibinfo{publisher}{Prentice-Hall}.

\bibitemdeclare{article}{PRS63}
\bibitem{PRS63}
\bibinfo{author}{M.~\surnamestart Perles\surnameend}, \bibinfo{author}{M.~O.
  \surnamestart Rabin\surnameend} \& \bibinfo{author}{E.~\surnamestart
  Shamir\surnameend} (\bibinfo{year}{1963}): \emph{\bibinfo{title}{The Theory
  of Definite Automata}}.
\newblock {\sl \bibinfo{journal}{IEEE Transactions on Electronic Computers}}
  \bibinfo{volume}{EC-12}(\bibinfo{number}{3}), pp. \bibinfo{pages}{233--243}.
\doi{10.1109/PGEC.1963.263534}

\bibitemdeclare{article}{Ro67}
\bibitem{Ro67}
\bibinfo{author}{A.~L. \surnamestart Rosenberg\surnameend}
  (\bibinfo{year}{1967}): \emph{\bibinfo{title}{A Machine Realization of the
  Linear Context-Free Languages}}.
\newblock {\sl \bibinfo{journal}{Information and Control}}
  \bibinfo{volume}{10}, pp. \bibinfo{pages}{175--188}.
\doi{10.1016/S0019-9958(67)80006-8}

\bibitemdeclare{article}{ShTh74}
\bibitem{ShTh74}
\bibinfo{author}{H.-J. \surnamestart Shyr\surnameend} \&
  \bibinfo{author}{G.~\surnamestart Thierrin\surnameend}
  (\bibinfo{year}{1974}): \emph{\bibinfo{title}{Ordered Automata and Associated
  Languages}}.
\newblock {\sl \bibinfo{journal}{Tamkang Journal of Mathematics}}
  \bibinfo{volume}{5}(\bibinfo{number}{1}).

\bibitemdeclare{inproceedings}{Si75}
\bibitem{Si75}
\bibinfo{author}{I.~\surnamestart Simon\surnameend} (\bibinfo{year}{1975}):
  \emph{\bibinfo{title}{Piecewise Testable Events}}.
\newblock In \bibinfo{editor}{H.~\surnamestart Brakhage\surnameend}, editor:
  {\sl \bibinfo{booktitle}{Proc.\ of the $2$nd {GI} Conference on Automata
  Theory and Formal Languages}}, {\sl
  \bibinfo{series}{LNCS}}~\bibinfo{volume}{33}, \bibinfo{publisher}{Springer},
  \bibinfo{address}{Kaiserslautern, Germany}, pp. \bibinfo{pages}{214--222}.

\bibitemdeclare{article}{Th68a}
\bibitem{Th68a}
\bibinfo{author}{G.~\surnamestart Thierrin\surnameend} (\bibinfo{year}{1968}):
  \emph{\bibinfo{title}{Permutation Automata}}.
\newblock {\sl \bibinfo{journal}{Mathematical Systems Theory}}
  \bibinfo{volume}{2}(\bibinfo{number}{1}), pp. \bibinfo{pages}{83--90}.
\doi{10.1007/BF01691347}

\end{thebibliography}
\end{document}